\documentclass[12pt, a4paper]{article}

\usepackage{amsfonts, mathrsfs, amsmath, amssymb, amsthm}

\usepackage{graphicx}
\usepackage[colorlinks=true, allcolors=blue]{hyperref}
\usepackage{xcolor}
\usepackage{mathabx}

\usepackage[utf8]{inputenc}
\usepackage[T1]{fontenc}
\usepackage{mathtools}                   
\usepackage{physics}                     
\usepackage{bm}                          
\usepackage{float}                       
\usepackage{subcaption}                  
\usepackage{booktabs}                    
\usepackage{hyperref}                    
\usepackage{cleveref}                    
\usepackage{geometry}                    
\usepackage{setspace}                    
\usepackage{enumitem}                    
\usepackage{doi}                         
\usepackage{cite}                        
\usepackage{microtype}                   
\usepackage{appendix}                    

\hypersetup{
    colorlinks = true,
    linkcolor  = blue!70!black,
    citecolor  = green!50!black,
    urlcolor   = cyan!70!black,
    pdftitle   = {Your Paper Title},
    pdfauthor  = {Author Name}
}

\theoremstyle{plain}
\newtheorem{theorem}{Theorem}[section]

\newtheorem{lemma}[theorem]{Lemma}

\theoremstyle{definition}
\newtheorem{definition}[theorem]{Definition}

\theoremstyle{remark}

\newcommand{\caI}{\mathcal I}

\newcommand{\caP}{\mathcal P}

\newcommand{\Z}{\mathbb{Z}}
\newcommand{\E}{\mathbb{E}}

\newcommand{\bbL}{\mathbb{L}}

\begin{document}

\title{\textbf{Order by disorder up to arbitrarily high temperature}}

\author{
      Ravish Mehta\thanks{\texttt{ravish.mehta@kuleuven.be}}\\
      \small Instituut voor Theoretische Fysica, KU Leuven, Belgium
  }

\date{\today}

\maketitle

\begin{abstract}
    \noindent
    We prove that a class of classical lattice models on $\mathbb{Z}^d$ ($d \geq 2$) with on-site space $\mathbb{N}_0$ and nearest neighbour interaction,
    exhibits long-range checkerboard order at sufficiently high temperature.  The ordering mechanism is purely entropic.
    The class of models contains the recently introduced  model of
    Han--Huang--Komargodski--Lucas--Popov \cite{HanEtAl2026}, by which our work is inspired.  The
    proof uses Pirogov--Sinai theory and the key input is a Peierls bound.

    \vspace{0.5em}
    \noindent\textbf{Keywords:} statistical mechanics, phase transitions
\end{abstract}


\section{Introduction}
\label{sec:introduction}

In statistical physics, ordered phases are typically associated with low temperature:
as the temperature increases, thermal fluctuations wash out long-range correlations. This expectation is supported by classical no-go results due to
Dobrushin~\cite{Dobrushin} and K\"unsch~\cite{Kunsch, Georgii2011}, which guarantee uniqueness of
the Gibbs measure at sufficiently high temperature whenever the local interactions are
bounded. A counterpoint has emerged in recent years. When the local single-spin space
is unbounded --- as for bosonic occupation numbers $n_x \in \mathbb{N}_0$ --- the
assumptions underlying the classical no-go theorems fail, and the possibility of
\emph{entropic order} arises: long-range order that not only survives at high temperature
but is driven by it. The mechanism is that a subset of configurations admits much
larger thermal fluctuations of an auxiliary degree of freedom than the rest, so that the
entropy gained outweighs any energetic preference for disorder. This is in the spirit of
\emph{order by disorder}, introduced by Villain, Bidaux, Carton, and Conte~\cite{Villain1980}
and independently by Shender~\cite{Shender1982, Henley1989} for frustrated magnets, though the regime
is opposite: order by disorder is a low-temperature effect, whereas entropic order arises
in the $\beta \to 0$ limit and selects configurations that are not energy minimisers.

Explicit models exhibiting entropic order have been studied in a series of recent
works~\cite{ChaiEtAl2022, KomargodskiPopov2024, HanEtAl2026, HuangEtAl2025}. In
particular, Han, Huang, Komargodski, Lucas, and Popov~\cite{HanEtAl2026} introduced the
classical lattice boson model with formal Hamiltonian
\begin{equation}
    H \;=\; U\!\!\sum_{\substack{x, y \in \mathbb{Z}^d \\ x \sim y}}\!\! n_x^p n_y^p
    \;+\; \mu \sum_{x \in \mathbb{Z}^d} n_x, \qquad n_x \in \mathbb{N}_0,
    \label{eq:HHKLP-Hamiltonian}
\end{equation}
and argued, via mean-field theory and Monte Carlo simulations, that this model exhibits long-range checkerboard order at all
sufficiently high temperatures when $p > 1$. 
The aim of the present paper is to prove this, using 
\emph{Pirogov--Sinai theory}~\cite{PirogovSinai1975,
PirogovSinai1976, Zahradnik1984, Sinai1982, BorgsImbrie1989, FriedliVelenik} --- including a full contour formalism,
a convergent cluster expansion, and a Peierls estimate. 

As a slight generalization of \eqref{eq:HHKLP-Hamiltonian}, we consider the formal Hamiltonian
\begin{equation}
    H \;=\; \sum_{x \in \mathbb{Z}^d} n_x \;+\;
    \!\!\sum_{\substack{x, y \in \mathbb{Z}^d \\ x \sim y}}\!\! f(n_x, n_y),
    \label{eq:our-Hamiltonian}
\end{equation}
and we impose four conditions on $f$, see Section~\ref{sec:system} that are in particular satisfied for original model~\eqref{eq:HHKLP-Hamiltonian} upon multiplying the Hamiltonian by a scale factor. 

Our main result (Theorem~\ref{thm:main}) states that for sufficiently small inverse temperature $\beta$, there are at least two distinct infinite-volume Gibbs states, corresponding to the two ways to color the lattice into a checkerboard, as will be explained in  Section~\ref{sec:system}.


While we were preparing this paper, another paper appeared on arxiv achieving a very similar goal: 
Andriolo, Nguyen, Richards, and
Sulejmanpasic~\cite{AndrioloEtAl2025} gave a rigorous proof of  ordering for \eqref{eq:HHKLP-Hamiltonian} with $p > 1$
on any bipartite lattice in dimension $d \geq 2$, using a Peierls-type domain-wall
argument combined with a rigorous asymptotic formula for the high-temperature cluster
partition function. Our paper offers a  different perspective on the problem, as we reduce the model to one that can be treated by Pirogov-Sinai theory. 


\section{Model and result}\label{sec:system}

\subsection{Lattice and configurations}

Fix an integer $d \geq 2$. We work on the hypercubic lattice $\mathbb{Z}^d$ equipped with the nearest-neighbour graph structure: $x \sim y$ if and only if $|x - y|_1 = 1$. The lattice is bipartite; we fix its natural two-colouring
\begin{equation}
    \mathbb{Z}^d \;=\; \mathcal{E}^1 \sqcup \mathcal{E}^2, \qquad  \mathcal{E}^1\,:=\, \{x \in \mathbb{Z}^d : x_1 + \cdots + x_d \text{ even}\}, \qquad \mathcal{E}^2 \,:=\, \mathbb{Z}^d \setminus \mathcal{E}^1.
\end{equation}
For any $S \subseteq \mathbb{Z}^d$, a set of occupation numbers on $S$ is an element $n_S = (n_x)_{x \in S} \in \mathbb{N}_0^S$, where $\mathbb{N}_0 := \{0, 1, 2, \ldots\}$. For each $n_S$, we can define the configuration $\omega_S=\omega_S(n_S) \in \{0,1\}^S$ as
\begin{equation}
    \omega_x \,:=\, \mathbf{1}[n_x \geq 1], \qquad x \in S.
\end{equation}
We denote by $d_\infty$ the $\ell^\infty$ metric on $\mathbb{Z}^d$.

The two checkerboard configurations $\omega^1, \omega^2 \in \{0,1\}^{\mathbb{Z}^d}$ are 
\begin{equation}
    \omega^1_x \,:=\, \mathbf{1}[x \in \mathcal{E}^1], \qquad \omega^2_x \,:=\, \mathbf{1}[x \in \mathcal{E}^2] \,=\, 1 - \omega^1_x.
\end{equation}
 We refer to them as the \emph{reference configurations} of type $\# \in \{1,2\}$. The translation $x \mapsto x + e_1$ exchanges $\omega^1 \leftrightarrow \omega^2$, realising a $\mathbb{Z}_2$ symmetry that we will exploit in Section~\ref{sec:cluster-expansion}.
 For later use, we also define \emph{blocks} of the form
 \begin{equation}\label{eq: blocks}
    B=B_i \;:=\; 2i + \{0,1\}^d, \qquad \text{for some $ i \in \mathbb{Z}^d$},
\end{equation}
Any region $\Lambda\subset\Z^d$ that is a union of blocks is called a \emph{block region}. 
\subsection{Finite-volume Hamiltonians and Gibbs measures}

Fix a function
\begin{equation}
    f : \mathbb{N}_0 \times \mathbb{N}_0 \;\to\; [0, \infty)
\end{equation}
satisfying the following assumptions, for all $m, n \in \mathbb{N}_0$:
\begin{enumerate}
    \item[(F1)] \emph{Symmetry:} $f(m,n) = f(n,m)$.
    \item[(F2)] \emph{Vanishing on empty sites:} $f(0,n) = 0$.
    \item[(F3)] \emph{Strict penalty on occupied pairs:} $f(m,n) > 0$ whenever $m, n \geq 1$.
    \item[(F4)] $N(T) = \#\bigl\{(m,n) \in \mathbb{N}^2 : (m+n)+f(m,n) \leq T\bigr\}=O(T^\alpha) $ for some $\alpha\in [0,1)$
\end{enumerate}

We define the Hamiltonian on a finite volume $\Lambda \subset\Z^d$ as
\begin{equation}
    H_\Lambda(n_\Lambda)=\sum_{x\in \Lambda}n_x+\sum_{\substack{x,y\in \Lambda\\ x\sim y}}f(n_x,n_y)
\end{equation}
To introduce the boundary conditions, we define, for any finite region $\Lambda$, the \emph{collar} of $\Lambda$ as
\[
    \Lambda^{\mathrm{col}} \;:=\; 
    \bigl\{ x \in \Lambda : d_\infty(x,\, \mathbb{Z}^d \setminus \Lambda) \leq 6 \bigr\}.
\]
The  finite-volume Gibbs measures $\mu^\#_{\Lambda,\beta}$ that we consider are defined as 
\begin{equation}\label{eq:measure}
\mu^\#_{\Lambda,\beta}(n_\Lambda) \;:=  \frac{1}{Z^\#_{\Lambda,\beta}} \;
\begin{cases}
    e^{-\beta H_\Lambda(n_\Lambda)} & \text{if } \quad \omega^{}_{\Lambda^{\text{col}}}(n_\Lambda)=\omega^\#_{\Lambda^{\text{col}}} \\
    0 & \text{otherwise}
\end{cases}.
\end{equation}
where $\omega^{}_{\Lambda^{\text{col}}}(n_\Lambda)$ denotes the configuration in $\Lambda^{\text{col}}$ derived from $n_\Lambda$, and where $Z^\#_{\Lambda,\beta}$ is a normalizing factor that makes $\mu^\#_{\Lambda,\beta}$ into a probability measure.

Our main result is 

\begin{theorem}
\label{thm:main}
Under assumptions (F1)--(F4), there exist a function
$\epsilon : (0,\infty) \to (0,\infty)$ with $\epsilon(\beta) \to 0$ as
$\beta \to 0$, such that for every $\beta >0$, every
$\# \in \{1,2\}$, every finite block region $\Lambda \subset \mathbb{Z}^d$, and every
$x \in \Lambda$,
\begin{equation}
\bigl| \E^\#_{\Lambda,\beta}
(\omega_x) - \omega^\#_x\bigr| \leq \epsilon(\beta),
\label{eq:main}
\end{equation}
uniformly in $\Lambda$ and $x$, and with $\E^\#_{\Lambda,\beta}$ the expectation w.r.t.\ to $\mu^\#_{\Lambda,\beta}$.  In particular, the infinite-volume limit points of the measures $\mu^1_{\Lambda,\beta}, \mu^2_{\Lambda,\beta}$ are distinct for sufficiently small $\beta$.
\end{theorem} 
As announced, the $\omega$ configurations singled out by our boundary conditions, concentrate on the checkerboard configurations $\omega^\#$ introduced above. 
The rest of the paper is devoted to the proof of this result.


\section{Contours}
\label{sec:contours}

To analyse fluctuations around the reference configurations $\omega^1, \omega^2$, we set up a Pirogov--Sinai contour formalism on a lattice consisting of the blocks defined in \eqref{eq: blocks}.   On this block-lattice the two reference configurations become constant, and the standard contour machinery for systems with finitely many translation-invariant ground states applies.

\subsection{The block-lattice}
\label{subsec:macro-lattice}
The blocks $B$ from \eqref{eq: blocks} give a partition of any block region. We view these blocks as sites of a 
 \emph{block-lattice} that we denote by $\mathbb{L}$ and we keep denoting its elements by $B$. 
  We equip $\mathbb{L}$ with the $\ell^\infty$ metric, which we still denote $d_\infty$. 


The reference configurations $\omega^1, \omega^2 \in \{0,1\}^{\mathbb{Z}^d}$ of Section~\ref{sec:system} are such that they are constant and distinct when viewed as functions of the blocks in $\bbL$.  




We say that a block $B \in \bbL$  is \emph{$\#$-correct in the configuration $\omega$} if
\begin{equation}
    \omega_{B'} \,=\,\omega^\#_{B'} \qquad \text{for every block $B'$ such that $d_\infty(B,B')\leq 1 $}.
\end{equation}
A block is \emph{correct} if it is $\#$-correct for some $\# \in \{1,2\}$, and \emph{incorrect} otherwise. The set of incorrect blocks in $\omega $ is denoted by $ I(\omega)$.

\begin{definition}[Thickened set of incorrect blocks]
\label{def:boundary}
The thickened incorrect set of $\omega$ is defined as
\begin{equation}
    \mathcal{I}(\omega) \;:=\; \{ B \in \bbL: d_{\infty}(B,I(\omega))\leq 1\}.
\end{equation}
\end{definition}

The following statement is crucial for what follows. It is an immediate consequence of the above definitions. 
\begin{figure}[h]
    \centering
    \includegraphics[width=\textwidth]{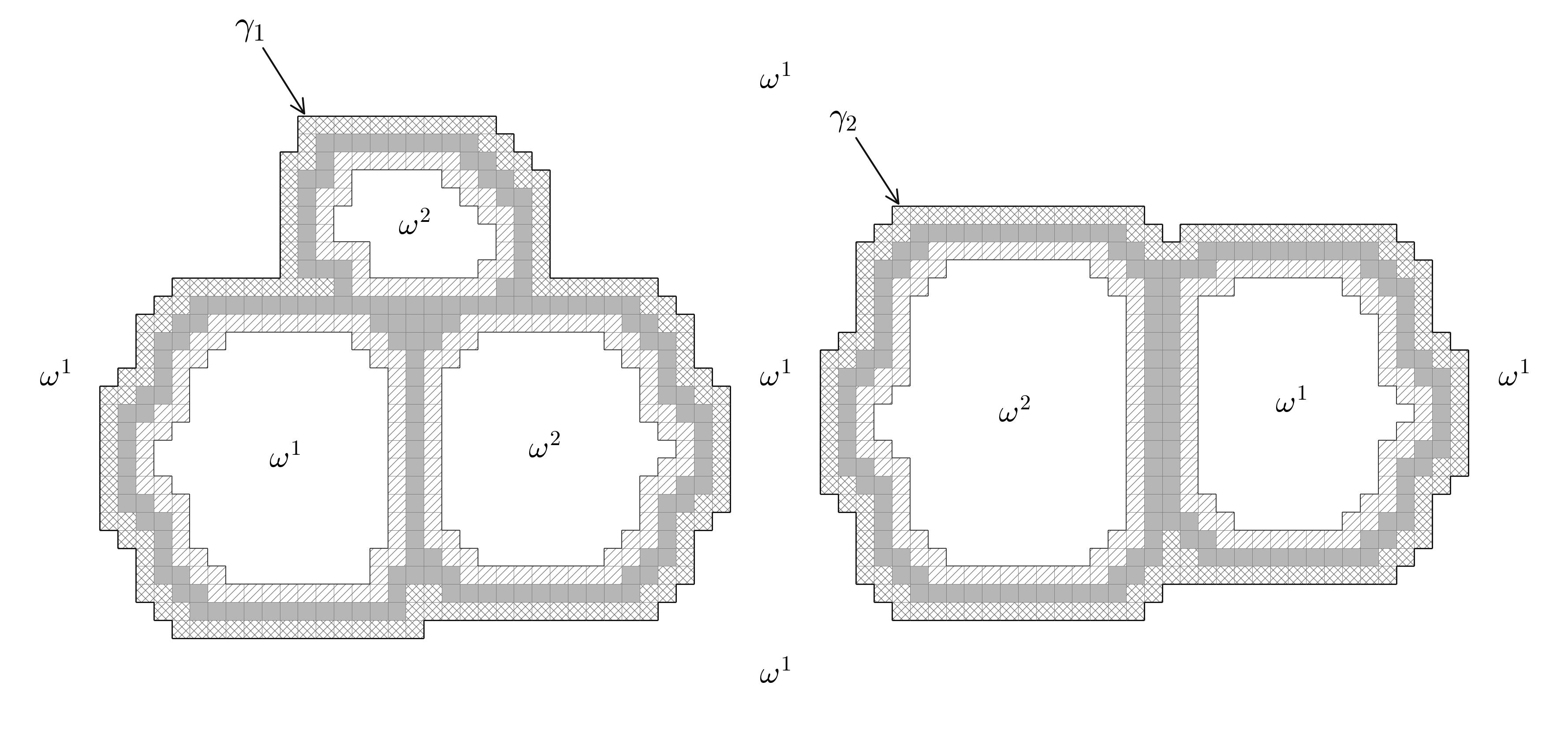}
    \caption{A configuration with two contours $\gamma_1$ and $\gamma_2$ with the set of incorrect blocks $I(\omega)$ denoted by the dark grey region and the patterned regions representing boundaries of the contours. The outer boundary has the configuration $\omega^1$ and the interior boundary has the same configuration as that of the interior components along which they are drawn.}
    \label{fig:mesh1}
\end{figure}
\begin{lemma}
\label{lem:correct-type}
Let $C \subseteq \mathbb{L} \setminus \mathcal{I}(\omega)$ be non-empty and $d_\infty$-connected. Then every block of $C$ is correct, and all blocks of $C$ share the same type.
\end{lemma}

\subsection{Contours}
\label{subsec:contours}
From now onwards, we restrict attention to configurations $\omega$ such that $\caI(\omega)\subset \bbL$ is finite.  We list a few definitions.

\begin{definition}[Contours]
\label{def:contour}
For each maximal $d_\infty$-connected component $\bar\gamma$ of $\mathcal{I}(\omega)$, the associated \emph{contour} is the pair
\begin{equation}\label{def: contour}
    \gamma \,:=\, (\bar\gamma,\, \omega_{\bar\gamma}).
\end{equation}
We call $\bar\gamma$ the \emph{support} of $\gamma$. 
We write $\Gamma(\omega)$ for the set of contours $\gamma$ determined in this way. 
\end{definition}
We will  use the term contour for $\gamma$ of the form \eqref{def: contour} with $\bar\gamma$ connected, even if it is not (explicitly) derived from a global $\omega$.

For a finite $A \subset \mathbb{L}$,  we define the exterior boundary as
$$
\partial^{\mathrm{ex}} A := \{B \in \mathbb{L} \setminus A : d_\infty(B, A) = 1\}.
$$

\begin{definition}[Exterior and type]
\label{def:interior-exterior-type}
Let $\gamma = (\bar\gamma, \omega_{\bar\gamma})$ be a contour (with $\bar\gamma$ finite). The complement $\mathbb{L} \setminus \bar\gamma$ decomposes into finitely many maximal $d_\infty$-connected components; exactly one of them is unbounded, and we denote it by $A_0$. We denote the remaining (bounded) components by $A_1, \ldots, A_k$, and set
\begin{equation}
    \mathrm{ext}\,\gamma \,:=\, A_0, \qquad \mathrm{int}\,\gamma \,:=\, \bigsqcup_{j=1}^k A_j.
\end{equation}

The boundary set $\partial^{\mathrm{ex}} A_j$ is $d_\infty$-connected for every $j \in \{0, 1, \ldots, k\}$, and there exists a unique type $\#$ for each $A_j$, denoted by $\#(A_j)$ such that 
\begin{equation}
    \omega_B \;=\; \omega^{\#}_B\qquad \text{for every } B \in \partial^{\mathrm{ex}} A_j.
\end{equation}
Note that $\partial^{\mathrm{ex}} A_j \subseteq \bar\gamma$ for every $j$, so the types are determined by the contour data $\omega_{\bar\gamma}$.

The interior of $\gamma$ refines into the disjoint union of its labelled parts:
\begin{equation}
    \mathrm{int}\,\gamma \,=\, \mathrm{int}_1\,\gamma \,\sqcup\, \mathrm{int}_2\,\gamma, \qquad \mathrm{int}_\#\,\gamma \,:=\, \bigsqcup_{j:  \#(A_j)=\#} A_j.
\end{equation}
The {type} of $\gamma$ is defined as
\begin{equation}
   \#(\gamma) \,:=\, \#(\mathrm{ext}\,\gamma) \,\in\, \{1, 2\}.
\end{equation}
\end{definition}
We note that every  contour in $\Gamma_{\mathrm{ext}}(\omega)$ has the same type $\#$, that can depend on $\omega$.

\begin{definition}[External contour]
\label{def:external}
Let $\Gamma(\omega)$ be the set of contours derived from $\omega$. A contour $\gamma \in \Gamma(\omega)$ is \emph{external} if there is no $\gamma' \in \Gamma(\omega) \setminus \{\gamma\}$ with $\bar\gamma \subseteq \mathrm{int}\,\gamma'$. We denote the set of external contours for $\omega$ by $\Gamma_{\mathrm{ext}}(\omega) \subseteq \Gamma(\omega)$ and , for a block region $\Lambda$,
\begin{equation}
\Lambda_{\mathrm{ext}}(\omega) := \Lambda \cap\left( \bigcap_{\gamma \in \Gamma_{\mathrm{ext}}(\omega)} \mathrm{ext}\gamma\right).
\end{equation}

\end{definition}

\begin{definition}[Compatibility of same-type contours]
\label{def:compatibility}
Two contours $\gamma_1, \gamma_2$ of the \emph{same type} are \emph{compatible}, written $\gamma_1 \sim \gamma_2$, if their supports satisfy
\begin{equation}
    d_\infty(\bar\gamma_1, \bar\gamma_2) \,>\, 1.
\end{equation}
A family $\Gamma$ of contours of the same type is \emph{compatible} if every pair of distinct elements of $\Gamma$ is compatible. 
\end{definition}

Note that  a mutually compatible family $\Gamma$ is not necessarily realizable  as $\Gamma(\omega)$ for some $\omega$.
In contrast, if a family of external contours $\Gamma_{\mathrm{ext}}$ is compatible, then it is can be realized as $\Gamma_{\mathrm{ext}}(\omega)$ for some $\omega$.

\section{Partition functions}\label{sec:partitionfunc}
\subsection{Polymer representation of partition function}
For any block region $\Lambda$ and configuration $\omega$, we define the partition functions
\begin{equation}
    Z_\Lambda(\omega) := \sum_{n_\Lambda: \, \omega_\Lambda(n_\Lambda)=\omega_\Lambda} e^{-\beta H_\Lambda(n_\Lambda)},
\end{equation}
and 
\begin{equation}
    Z^\#_\Lambda(\omega) := \begin{cases} Z_\Lambda(\omega) & \text{if} \, \omega_{\Lambda^{\text{col}}}=\omega^\#_{\Lambda^{\text{col}}}  \\
    0 & \text{otherwise}
    \end{cases},
\end{equation}
and we note that
$$
  Z^\#_\Lambda =\sum_{\omega_\Lambda \in \{0,1\}^{\Lambda}}   Z^\#_\Lambda(\omega_\Lambda) 
$$
where $Z^\#_\Lambda$ was the normalizing factor in \eqref{eq:measure}. Finally, we also
 define the partition function over the reference configuration as
\begin{equation}
    Z^0_\Lambda:=Z_\Lambda(\omega^\#).
\end{equation}
and, since $\Lambda$ is a block region, we can check that 
$Z^0_\Lambda$ is independent of $\#$, see Section \ref{sec:system} for an explicit expression.  

\begin{lemma}[Factorization across external contours]
\label{lem:factorization}
For any configuration $\omega$,
\begin{equation}
   Z^{\#}_{\Lambda}(\omega) \;=\; Z^0_{\Lambda_{\mathrm{ext}}}(\omega)\,
   \prod_{\gamma \in \Gamma_{\mathrm{ext}}(\omega) } Z_{\overline{\gamma}}(\omega_{\bar \gamma}) \prod_{j=1}^{k} Z^{\#(A_j)}_{A_j}(\omega).
\end{equation}
\end{lemma}

\begin{proof}
Each of these partition functions is a sum over occupation numbers $n_\Lambda$ compatible with the given configuration $\omega$. The crucial insight is that the definition of the contours $\gamma$ allows us to drop the interaction terms in the Hamiltonian between sites in $\overline{\gamma}$ and $\overline{\gamma}^c$. Therefore the partition function factorizes over these regions.
\end{proof}

We now reorganize this sum by summing first over compatible families $\Gamma_{\mathrm{ext}}$ of external contours, and then over all configurations $\omega_{\Lambda}$ such that   $\Gamma_{\mathrm{ext}}(\omega)=\Gamma_{\mathrm{ext}}$. 
This yields
\begin{equation}
Z^\#_\Lambda=\sum_{\Gamma_{\mathrm{ext}}}Z^0_{\Lambda_{\mathrm{ext}}}\prod_{\gamma\in \Gamma_{\mathrm{ext}}} Z_{\overline{\gamma}}(\gamma) \prod_{j=1}^{k} Z^{\#(A_j)}_{A_j},
\end{equation}
where we abbreviated $Z_{\overline\gamma}(\gamma)=Z_{\overline\gamma}(\omega_{\overline\gamma})$.
Dividing by $Z^0_\Lambda$ and using the analogous factorization of the reference partition function,
\begin{equation}
    \frac{Z^{\#}_{\Lambda}}{Z_\Lambda^0} \;=\; \sum_{\Gamma_{\mathrm{ext}}} \prod_{\gamma \in \Gamma_{\mathrm{ext}}} \frac{Z_{\overline{\gamma}}(\gamma)}{Z^0_{\overline{\gamma}}} \prod_{j=1}^{k} \frac{Z^{\#(A_j)}_{A_j}}{Z^0_{A_j}}.
\end{equation}
Denoting $\frac{Z^{\#}_{\Lambda}}{Z_\Lambda^0}$ by $\Xi^\#_\Lambda$, this reads
\begin{equation}
     \Xi^\#_\Lambda\;=\; \sum_{\Gamma_{\mathrm{ext}}} \prod_{\gamma \in \Gamma_{\mathrm{ext}}} \frac{Z_{\overline{\gamma}}(\gamma)}{Z^0_{\overline{\gamma}}} \prod_{j=1}^{k} \Xi^{\#(A_j)}_{A_j}.
\end{equation}
Multiplying and dividing each interior factor by $\Xi^{\#}_{A_j}$ yields
\begin{equation}
    \Xi^\#_\Lambda \;=\; \sum_{\Gamma_{\mathrm{ext}}} \prod_{\gamma \in \Gamma_{\mathrm{ext}}} w^\#(\gamma) \prod_{j=1}^{k} \Xi^\#_{A_j},
\end{equation}
where the contour weight is
\begin{equation}
    w^\#(\gamma) \;:=\; \frac{Z_{\overline{\gamma}}(\gamma)}{Z^0_{\overline{\gamma}}} \prod_{j=1}^{k} \frac{\Xi^{\#(A_j)}_{A_j}}{\Xi^{\#}_{A_j}}.
\end{equation}
This sets up a recursion; iterating over all nesting levels collapses the right-hand side to a purely geometric polymer model
\begin{equation}\label{eq:polymer}
    \Xi^\#_\Lambda \;=\; \sum_{\Gamma\,\text{compatible}} \prod_{\gamma \in \Gamma} w^\#(\gamma),
\end{equation}
where we stressed in the sum that $\Gamma$ are compatible. 
\subsection{Dimer estimate}
First, we check that for any block region $\Lambda$
$$
 Z^0_\Lambda= (Z_{\mathrm{sing}})^{|\Lambda|/2}, \qquad   Z_{\mathrm{sing}} = \sum_{n \geq 1} e^{-\beta n} = \frac{1}{e^{\beta }-1}.
$$
For later use, we note that
\begin{equation}\label{eq:Zsingasymp}
    Z_{\mathrm{sing}}=\frac{1}{\beta }\bigl(1 + O(\beta)\bigr), \qquad \text{as}
\,  \beta\to 0  \end{equation}
By a dimer, we mean a pair of adjacent sites where the configuration is $1$. The resulting partition function is 
$$Z_{\mathrm{dim}} = \sum_{m,n \geq 1} e^{-\beta (m+n) - \beta f(m,n)}.$$

\begin{lemma}[Dimer excess]\label{lem:dimer-excess}
 For $\beta$ sufficiently small and for some constant $p$,
\begin{equation}\label{eq: claim lemma dimer}
    \frac{Z_{\mathrm{dim}}}{Z_{\mathrm{sing}}} \;\leq\; p\, \beta^{1-\alpha}.
\end{equation}
\end{lemma}

\begin{proof}
Let $E_{m,n} = (m+n) + f(m,n)$ and write $e^{-\beta E_{m,n}} = \int_{E_{m,n}}^{\infty} \beta e^{-\beta T}\,dT$. 
By exchanging the order of summation and integration,
\begin{equation}
    Z_{\mathrm{dim}} \;=\; \beta \int_0^\infty e^{-\beta T} \biggl( \sum_{\substack{m,n \ge 1 \\ E_{m,n} \le T}} 1 \biggr) dT \;=\; \beta \int_0^\infty e^{-\beta T} N(T)\, dT.
\end{equation}
From (F4) we have $N(T) \leq K T^\alpha$, hence
\begin{equation}
    Z_{\mathrm{dim}} \;\leq\; K \beta \int_0^\infty e^{-\beta T} T^\alpha \,dT \;=\; K \Gamma(\alpha+1)\, \beta^{-\alpha}.
\end{equation}
Combining this with \eqref{eq:Zsingasymp}, we get the claim \eqref{eq: claim lemma dimer}.
\end{proof}

\subsection{Bound on contour weights}
Given $\omega$, a block $B$ is  called a defect block if $\omega_B\neq \omega_B^\#$ for any $\#$.
\begin{lemma}\label{lem:small-piece}
Let $\omega$ be a configuration and $\beta$ sufficiently small.
\begin{enumerate}
    \item[\rm (i)] If $B \in \mathbb{L}$ is a defect block, then
    \begin{equation}\label{eq:per-block}
        Z_{B}(\omega) \;\leq\; Z_{\mathrm{sing}}^{2^{d-1}}\, \frac{Z_{\mathrm{dim}}}{Z_{\mathrm{sing}}}.
    \end{equation}
    \item[\rm (ii)] If $B, B' \in \mathbb{L}$ are adjacent blocks that are both correct but of opposite type, then
    \begin{equation}\label{eq:per-pair}
        Z_{B \cup B'}(\omega) \;\leq\; Z_{\mathrm{sing}}^{2^d}\, \frac{Z_{\mathrm{dim}}}{Z_{\mathrm{sing}}^2}.
    \end{equation}
\end{enumerate}
\end{lemma}

\begin{proof}
\textbf{(i)} Set $N := |\{x \in B : \omega_x = 1\}|$ and let $M$ be a \emph{matching} in the  "occupied" subgraph of $B$, i.e.\  $\{x \in B : \omega_x = 1\}$. A matching is a set of edges which do not share any vertex.  Dropping all interaction terms, except those  corresponding to edges in $M$, yields
\begin{equation}\label{eq:Bi-skeleton}
    Z_{B}(\omega) \;\leq\; Z_{\mathrm{sing}}^{N}\, \left(\frac{Z_{\mathrm{dim}}}{Z_{\mathrm{sing}}^2}\right)^{|M|}.
\end{equation}
Assuming that a block $B$ is a defect block, we can separate $B$ into 3 cases:
\\
\begin{figure}[h]
    \centering
    \includegraphics[width=\textwidth]{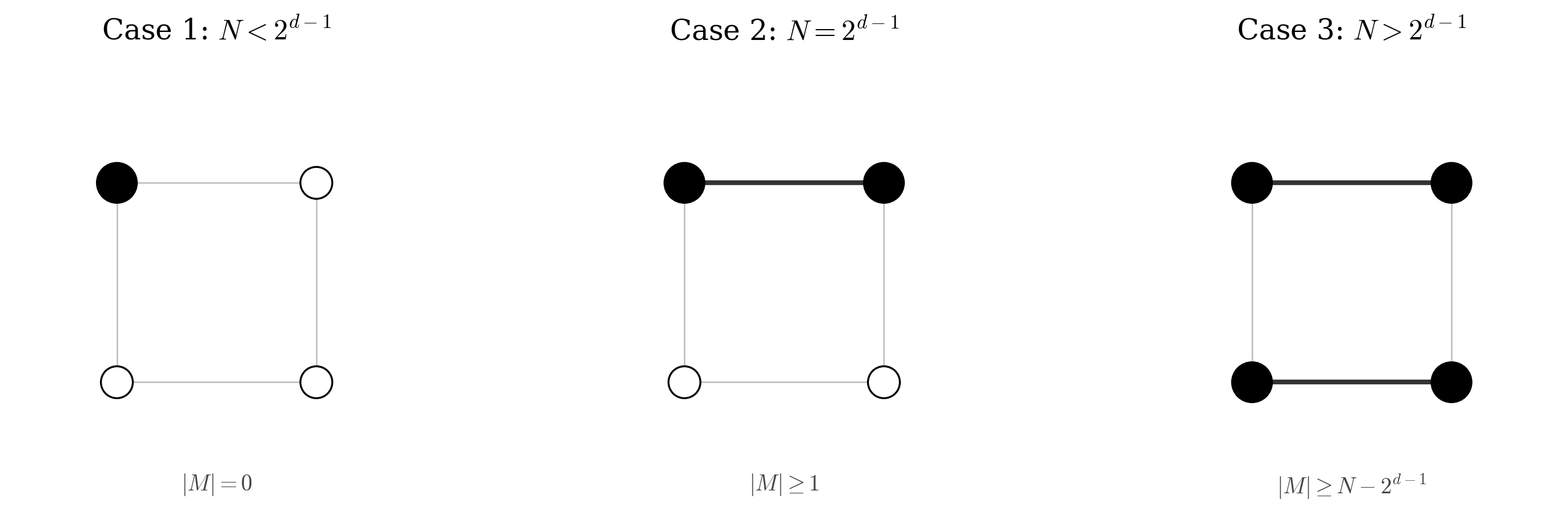}
    \caption{The figure depicts 3 cases of defect blocks for a $2\times 2$ block. The black vertices $x$ have $\omega_x=1$. The black edges are the elements of $M$. We get $|M|=1$ in case 2 and $|M|=2$ in the case 3 block.}
    \label{fig:lattice}
\end{figure}
\\
\emph{Case 1: $N < 2^{d-1}$.} Setting $|M| = 0$ in \eqref{eq:Bi-skeleton},
\begin{equation}
    \frac{Z_{B}(\omega)}{Z_{\mathrm{sing}}^{2^{d-1}}} \leq Z_{\mathrm{sing}}^{-1} \leq \frac{Z_{\mathrm{dim}}}{Z_{\mathrm{sing}}},
\end{equation}
the last bound using $Z_{\mathrm{dim}} \geq 1$ for $\beta$ small.

\emph{Case 2: $N = 2^{d-1}$.} As $B$ is a defect block, we can choose a non-empty matching $M$ s.t. $|M| \geq 1$, giving
\begin{equation}
\frac{Z_{B}(\omega)}{Z_{\mathrm{sing}}^{2^{d-1}}} \leq \frac{Z_{\mathrm{dim}}}{Z_{\mathrm{sing}}^2} \leq \frac{Z_{\mathrm{dim}}}{Z_{\mathrm{sing}}},
\end{equation}
using $Z_{\mathrm{sing}} \geq 1$ for $\beta$ small.

\emph{Case 3: $N > 2^{d-1}$.} We can find a matching $M$ with $|M| \geq N - 2^{d-1}$. Substituting in~\eqref{eq:Bi-skeleton},
\begin{equation}
    \frac{Z_{B}(\omega)}{Z_{\mathrm{sing}}^{2^{d-1}}} \;\leq\; \left(\frac{Z_{\mathrm{dim}}}{Z_{\mathrm{sing}}}\right)^{N-2^{d-1}} \;\leq\; \frac{Z_{\mathrm{dim}}}{Z_{\mathrm{sing}}},
\end{equation}
since $Z_{\mathrm{dim}}/Z_{\mathrm{sing}} \leq 1$ in the regime considered.

\textbf{(ii)}  Direct inspection of $\omega^1$ and $\omega^2$ on the shared face of $B$ and $B'$ shows at least one pair of adjacent occupied sites across the interface; retaining this edge\footnote{That is: retaining the interaction term corresponding to that edge} while dropping all other edges between $B$ and $B'$, yields the factor $Z_{\mathrm{dim}}$. Each block contributes $Z_{\mathrm{sing}}^{2^{d-1}-1}$ from its $2^{d-1}-1$ remaining correctly occupied sites. So,
\begin{equation}
    Z_{B \cup B'}(\omega) \;\leq\; (Z_{\mathrm{sing}}^{2^{d-1}-1})^2\, Z_{\mathrm{dim}}.
\end{equation}
\begin{figure}[h]
    \centering
    \includegraphics[width=0.6\textwidth]{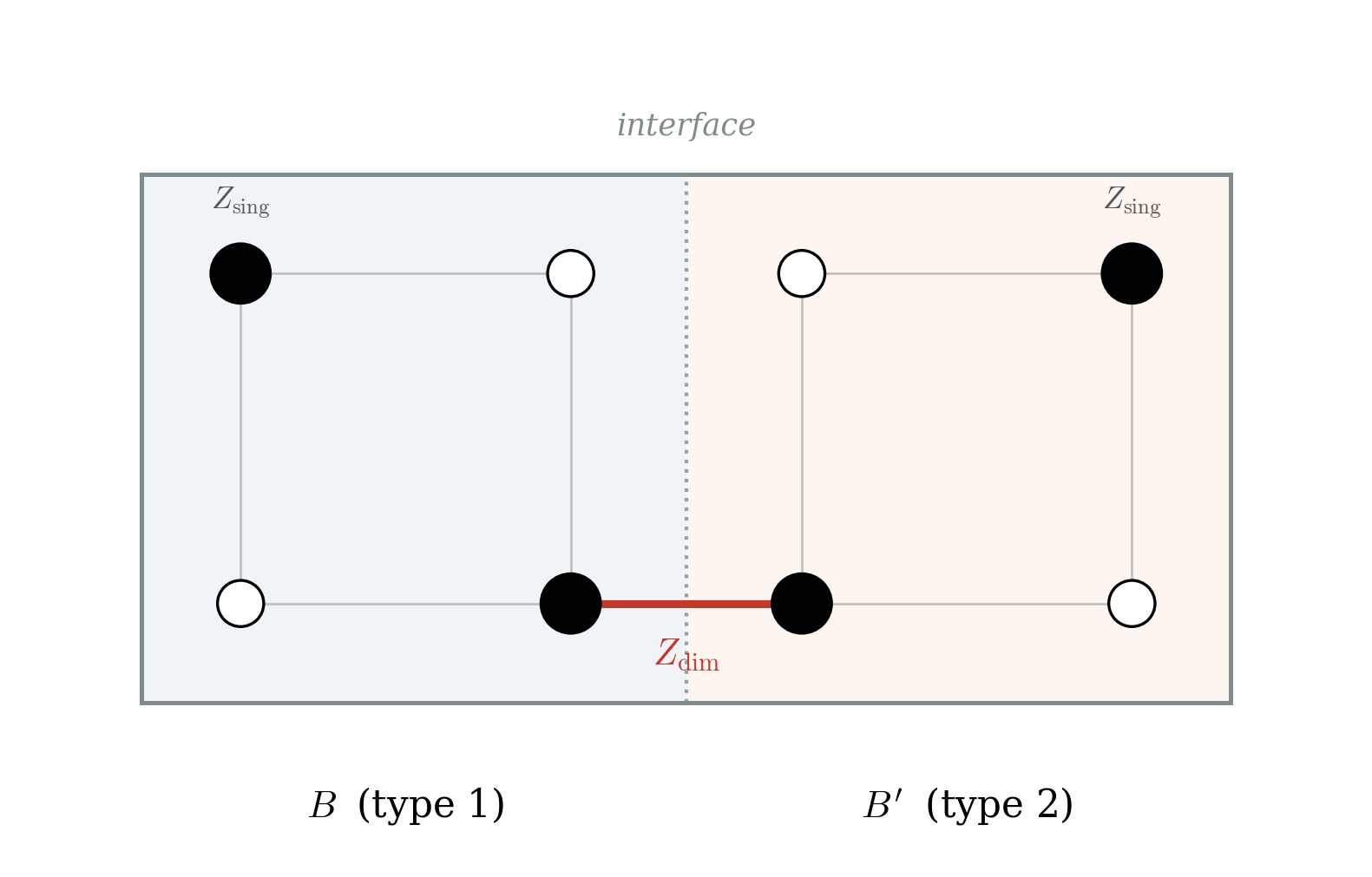}
    \caption{The two occupied sites contribute $Z_{\text{sing}}$ to the partition function and the two neighbouring sites connected by the red bond contribute $Z_{\text{dim}}$.} 
    \label{fig:lattice}
\end{figure}
\end{proof}

\begin{theorem}\label{thm:surface-energy}
For any contour $\gamma$, 
\begin{equation}
    Z_{\bar\gamma}(\gamma) \;\leq\; Z^0_{\bar\gamma}\, e^{-c(\beta)|\bar\gamma|},
\end{equation}
where $c(\beta) \to \infty$ as $\beta \to 0$.
\end{theorem}

\begin{proof}
For any partition $A = A_1 \sqcup ...\sqcup A_k$:
\begin{equation}\label{eq:bipartition}
    Z_A(\omega) \leq \prod_{i=1}^k Z_{A_i}(\omega)
\end{equation}
This follows by simply dropping interaction terms in the Hamiltonian that act between different $A_i$.

Next, let $D$ be the set of defect blocks in $\bar\gamma$. We let $\caP$ be a maximal (w.r.t.\ inclusion) collection of pairs $\{B,B'\}$ of adjacent correct blocks of opposite type, such that $B,B' \in \bar\gamma \setminus D$, and such that no block features twice in $\caP$.

We now consider a partition of $\bar\gamma$ into defect blocks in $D$, pairs in $\mathcal{P}$, and blocks in $\bar\gamma \setminus (D \cup V(\mathcal{P}))$ with $V(\mathcal{P})$ the set of blocks featuring in $\caP$.
Applying  \eqref{eq:bipartition} we get
\begin{equation}
    Z_{\bar\gamma}(\omega) \;\leq\; \prod_{B \in D} Z_B(\omega) \cdot \prod_{\{B,B'\} \in \mathcal{P}} Z_{B \cup B'}(\omega) \cdot \prod_{B \in \bar\gamma \setminus (D \cup V(\mathcal{P}))}  Z_B(\omega).
\end{equation}
Hence
\begin{equation}\label{eq:bound-product}
    \frac{Z_{\bar\gamma}(\omega)}{Z^0_{\bar\gamma}}
    \;\leq\; \left(\frac{Z_{\mathrm{dim}}}{Z_{\mathrm{sing}}}\right)^{|D|}
    \left(\frac{Z_{\mathrm{dim}}}{Z_{\mathrm{sing}}^2}\right)^{|\mathcal{P}|}
    \;\leq\; \left(\frac{Z_{\mathrm{dim}}}{Z_{\mathrm{sing}}}\right)^{|D| + |\mathcal{P}|},
\end{equation}
the last step using $Z_{\mathrm{sing}} \geq 1$ for $\beta$ small.

Every incorrect block $B\in I(\omega)\cap \bar\gamma$ satisfies $d_\infty(B,D \cup V(\mathcal{P}))\leq 1$. Since every  block $B$ in $\bar\gamma$ satisfies $d_\infty(B,I(\omega))\leq 1 $ by definition of $\mathcal{I}(\omega)$, it follows that every block $B$ in $\bar\gamma$ satisfies  $  d_\infty(B,D \cup V(\mathcal{P}))\leq 2$, and hence
\begin{equation}\label{eq:counting}
    |\bar\gamma| \;\leq\; 5^d (|D| + 2|\mathcal{P}|) \;\leq\; 2 \cdot 5^d \, (|D| + |\mathcal{P}|).
\end{equation}

Combining \eqref{eq:Zsingasymp}, \eqref{eq:bound-product}, \eqref{eq:counting}  and Lemma~\ref{lem:dimer-excess}, we obtain the claim with $c(\beta) := \frac{(1-\alpha)|\log\beta|-\log p}{2 \cdot 5^d}$. For $\beta$ small, $|\log\beta|\gg\log p$ and we get $Z_{\bar\gamma}(\gamma) \;\leq\; Z^0_{\bar\gamma}\, e^{-c\log\beta|\bar\gamma|}$ for some constant $c$.
\end{proof}

\section{Cluster expansion of the bulk partition function}
\label{sec:cluster-expansion}

\subsection{Bulk--boundary decomposition}
\label{subsec:bulk-boundary}
Let $\mathcal{C}^\#$ be the set of all contours of type $\#$ and $\mathrm{B}(k)=\{B\in \mathbb{L}:d_\infty(B,0)\leq k\}$. Earlier we had used $\Lambda\subset \Z^d$ to denote a block region. We will use the same symbol to also denote $\Lambda\subset \mathbb{L}$ with $\Lambda$ being a set of blocks now and $|\Lambda|$ counting the number of blocks.
\begin{theorem}[Bulk--boundary decomposition]
\label{thm:bulk-boundary}
Assume that $\forall \gamma\in \mathcal{C}^\#$, $w(\gamma)\leq e^{-\tau|\bar\gamma|}$ for $\tau$ sufficiently large.
Then the bulk free energy density
\begin{equation}
    g^\# \;:=\; \lim_{k\to\infty} \frac{1}{|\mathrm{B}(k)|}\,\log\Xi^\#(\mathrm{B}(k))
    \label{eq:g-definition}
\end{equation}
exists, and for every finite $\Lambda\subset\mathbb{L}$,
\begin{equation}\label{eq:bulk-boundary}
    \left|\log\Xi^\#(\Lambda) - g^\#|\Lambda|\right|\;\leq\; |\partial^{\mathrm{in}}\Lambda|.
\end{equation}
\end{theorem}

\begin{proof}
The polymer partition function $\Xi_\Lambda$ fits into the abstract framework
of~\cite{FriedliVelenik}. The key point is that the exponential weight bound
on contours implies, via the Koteck\'y--Preiss criterion, that the contribution
of contours crossing the boundary $\partial^{\mathrm{in}}\Lambda$ is exponentially
suppressed, which separates the bulk from the boundary. We refer
to~\cite[Thm.~7.28, Exercise~7.13]{FriedliVelenik} for the complete proof.
\end{proof}

\subsection{$\mathbb{Z}_2$ symmetry}
\label{subsec:Z2-symmetry}

\begin{lemma}
\label{lem:Z2}
We have $g^1 = g^2$.
\end{lemma}

\begin{proof}
Consider the lattice translation $T:\mathbb{Z}^d\to\mathbb{Z}^d$, $x\mapsto x+e_1$,
where $e_1$ is the first standard basis vector. Since $x_1+\cdots+x_d$ and
$(x_1+1)+x_2+\cdots+x_d$ have opposite parities, $T$ exchanges the sublattices
$\mathcal{E}^1\leftrightarrow \mathcal{E}^2$ and therefore maps $\omega^1\leftrightarrow\omega^2$.
Fix $n \geq 0$, and set the weight of all contours with $|\text{int}\,\gamma|> n$ to 0. We denote the restricted polymer model on this set of weights by $\Xi^\#_n(\Lambda)$.
From Theorem~\ref{thm:bulk-boundary}, we have 
\begin{equation}\label{eq:gsym1}
    \left|\log\Xi_n^1(\Lambda) - g^1|\Lambda|\right|\;\leq\; |\partial^{\mathrm{in}}\Lambda|
\end{equation}
and 
\begin{equation}\label{eq:gsym2}
    \left|\log\Xi_n^2(T\Lambda) - g^2|\Lambda|\right|\;\leq\; |\partial^{\mathrm{in}}\Lambda|.
\end{equation}

Note that
$$\Xi_n^1(\Lambda)=\Xi_n^2(T\Lambda).$$
$T\Lambda$ is not a block region anymore because we shift by one site but we can redefine our blocks appropriately for $T\Lambda$ and make it a block region so that $\Xi^2_n(T\Lambda)$ is well defined.
Subtracting eqs. \eqref{eq:gsym1} and \eqref{eq:gsym2}, we get
\begin{equation}
    \left|g^1|\Lambda|-g^2|\Lambda|\right|\;\leq\; 2|\partial^{\mathrm{in}}\Lambda|
\end{equation}
As $g^\#$ is independent of $\Lambda$, we can pick $\Lambda$ large enough to get $g^1=g^2$.
\end{proof}

We write $g$ for the common value henceforth.

\subsection{Inductive bound on contour weights}
\label{subsec:weight-bound}

\begin{theorem}
\label{thm:weight-bound}
There exists $\beta_0>0$ such that for all $\beta\in(0,\beta_0)$,
\begin{equation}
    w^\#(\gamma) \;\leq\; e^{-\tau(\beta)|\bar\gamma|}
    \label{eq:weight-bound}
\end{equation}
holds for every contour $\gamma$, where $\tau(\beta):=c|\log\beta|-2\to\infty$
as $\beta\to 0$, with $c>0$ the constant from Theorem~\ref{thm:surface-energy}.
\end{theorem}

\begin{proof}
We induct on the size of $\text{int}\gamma$.

\emph{Base case.} If $|\mathrm{int}\,\gamma|=0$, then $w^\#(\gamma)=\frac{Z_{\bar\gamma}(\gamma)}{Z^0_{\bar\gamma}}$
and the Peierls bound (Theorem~\ref{thm:surface-energy}) gives immediately
\begin{equation}
    w^\#(\gamma) \;\leq\; e^{-c|\log\beta||\bar\gamma|} \;=\; e^{-\tau(\beta)|\bar\gamma|-2|\bar\gamma|}
    \;\leq\; e^{-\tau(\beta)|\bar\gamma|}.
\end{equation}

\emph{Inductive step.} Suppose $w(\gamma')\leq e^{-\tau|\bar \gamma'|}$ for all contours with $|\text{int}\gamma'|\leq n$, and let $|\text{int}\gamma|=n+1$. Since every contour $\gamma'$ contributing to $\Xi^\#_n(\mathrm{int}_{\#'}\gamma)$ and $\Xi^{\#'}_n(\mathrm{int}_{\#'}\gamma)$ has $|\text{int}\gamma'|\leq n$, the induction hypothesis gives $w^\#(\gamma') \leq e^{-\tau|\bar\gamma'|}$ for every such $\gamma'$. So, we get
\begin{align}
\log \Xi^{\#'}_n(\mathrm{int}_{\#'}\gamma) &\leq g |\mathrm{int}_{\#'}\gamma| + |\partial^{\mathrm{in}}\mathrm{int}_{\#'}\gamma|,  \\
\log \Xi^\#_n(\mathrm{int}_{\#'}\gamma) &\geq g |\mathrm{int}_{\#'}\gamma| - |\partial^{\mathrm{in}}\mathrm{int}_{\#'}\gamma|, 
\end{align}
Subtracting,
\begin{equation}
\left|\log \frac{\Xi^{\#'}_n(\mathrm{int}_{\#'}\gamma)}{\Xi^\#_n(\mathrm{int}_{\#'}\gamma)}\right| \leq 2|\partial^{\mathrm{in}}\mathrm{int}_{\#'}\gamma| \leq 2|\bar\gamma|.
\end{equation}
Combining with the Peierls bound,
\begin{equation}
    \log w^\#(\gamma)
    \;\leq\; \log \frac{Z_{\bar\gamma}(\gamma)}{Z^0_{\bar\gamma}} + 2|\bar\gamma|
    \;\leq\; \bigl(-c|\log\beta| + 2\bigr)|\bar\gamma|.
\end{equation}
Choosing $\beta_0$ small enough that $c|\log\beta|-2>\tau_0$ for $\beta<\beta_0$
completes the induction, with $\tau(\beta)=c|\log\beta|-2$.
\end{proof}
We are now ready to give the 

\begin{proof}[Proof of Theorem \ref{thm:main}]

We fix
$\beta \in (0,\beta_0)$, so that the contour weight bound
\begin{equation}
w^\#(\gamma) \;\leq\; e^{-\tau(\beta)|\bar\gamma|}, \qquad
\tau(\beta) := c|\log\beta| - 2 \xrightarrow{\beta \to 0} \infty,
\label{eq:wbound-proof}
\end{equation}
holds for every contour $\gamma$ of type $\#$.
Fix $x \in \Lambda$ and let $B \in \mathbb{L}$ be the block containing $x$. Since
$\omega_x,\omega^\#_x \in \{0,1\}$,
\begin{equation}
\bigl|\E^\#_{\Lambda,\beta}(\omega_x) - \omega^\#_x\bigr|
\;\leq\; \mu^\#_{\Lambda,\beta}(\omega_x \neq \omega^\#_x).
\label{eq:reduction}
\end{equation}

If $\omega_x \neq \omega^\#_x$, then $B$ is either incorrect, in which case
$B \in \bar\gamma$ for some $\gamma \in \Gamma(\omega)$, or correct of type
$\#' \neq \#$. In the latter case, by Lemma~\ref{lem:correct-type} and the boundary condition on
$\Lambda^{\rm col}$, $B$ lies in a bounded maximal $d_\infty$-connected component of
$\mathbb{L} \setminus \caI(\omega)$ of type $\#'$, hence
$B \in \mathrm{int}_{\#'}\gamma$ for some $\gamma \in \Gamma(\omega)$.
Either way, $B \in \bar\gamma \cup \mathrm{int}_{\#'}\gamma$ for some
$\gamma \in \Gamma(\omega)$ of type $\#$. Using the bound $\mu^\#_{\Lambda,\beta}(\gamma \in \Gamma) \leq w^\#(\gamma)$,
\begin{equation}
\mu^\#_{\Lambda,\beta}(\omega_x \neq \omega^\#_x)
\;\leq\;
\sum_{\substack{\gamma \in \mathcal{C}^\# \\ B \in \bar\gamma\,\cup\,\mathrm{int}\,\gamma}}
w^\#(\gamma).
\label{eq:unionbound}
\end{equation}

If $B \in \bar\gamma \cup \mathrm{int}\,\gamma$ and $|\bar\gamma| = k$, then since
$\bar\gamma$ is $d_\infty$-connected of diameter at most $k$, there exists
$B' \in \bar\gamma$ with $d_\infty(B,B') \leq k$. Using the standard lattice bound
$\#\{\gamma \in \mathcal{C}^\# : B' \in \bar\gamma,\,|\bar\gamma| = k\} \leq e^{c_d k}$,
\[
\#\bigl\{\gamma \in \mathcal{C}^\# : B \in \bar\gamma \cup \mathrm{int}\,\gamma,\ |\bar\gamma| = k\bigr\}
\;\leq\; (2k+1)^d\, e^{c_d k} \;\leq\; e^{c'_d k},
\]
for some $c'_d > c_d$ depending only on $d$. Combining with
\eqref{eq:wbound-proof} and \eqref{eq:unionbound},
\begin{equation}
\mu^\#_{\Lambda,\beta}(\omega_x \neq \omega^\#_x)
\;\leq\; \sum_{k \geq 1} e^{c'_d k}\, e^{-\tau(\beta) k}
\;=\; \frac{e^{-(\tau(\beta) - c'_d)}}{1 - e^{-(\tau(\beta) - c'_d)}}
\;=:\; \epsilon(\beta),
\end{equation}
after shrinking $\beta_0$ to ensure $\tau(\beta) > c'_d$. Then
$\epsilon(\beta) \to 0$ as $\beta \to 0$, uniformly in $\Lambda$ and $x$,
establishing \eqref{eq:main}.

For the second claim, fix $\beta$ small enough that $\epsilon(\beta) < 1/2$, and
let $\mu^\#_\infty$ be any infinite-volume limit point of
$(\mu^\#_{\Lambda_n,\beta})$ along a sequence $\Lambda_n \uparrow \mathbb{Z}^d$ of
block regions. Passing \eqref{eq:main} to the limit gives
$|\E^\#_\infty(\omega_x) - \omega^\#_x| \leq \epsilon(\beta)$ for every $x$.
Since $\omega^1_x \neq \omega^2_x$,
\[
\bigl|\E^1_\infty(\omega_x) - \E^2_\infty(\omega_x)\bigr|
\;\geq\; 1 - 2\epsilon(\beta) \;>\; 0,
\]
so every limit point of $(\mu^1_{\Lambda_n,\beta})$ differs from every limit point
of $(\mu^2_{\Lambda_n,\beta})$.
\end{proof}
\noindent\textbf{Acknowledgements} The author thanks Andrew Lucas for suggesting this problem.\\
\textbf{Funding} This work was supported by FWO and F.R.S.-FNRS under the Excellence of Science (EOS) programme through the research project G0H1122N EOS 40007526 CHEQS, by the KU Leuven Runner-up Grant iBOF/DOA/20/011, and by the KU Leuven internal grant C14/21/086.\\

\bibliographystyle{unsrt} 
\bibliography{references}

\end{document}